\documentclass[a4paper,fleqn]{cas-dc}

\usepackage[numbers,sort&compress]{natbib}
\usepackage{amsmath,amssymb,amsfonts}
\usepackage{algorithmic}
\usepackage{algorithm}
\usepackage{graphicx}
\usepackage{textcomp}
\usepackage{upgreek}
\usepackage{bm}

\usepackage{amsthm}
\usepackage[font=footnotesize]{caption}
\usepackage{subcaption}

\usepackage{xcolor}

\usepackage{makecell}
\usepackage{booktabs}
\usepackage{multirow}

\usepackage{array}
\usepackage{censor}

\newif\ifcensoring

\newcommand{\RNum}[1]{\uppercase\expandafter{\romannumeral #1\relax}}

\newtheorem{theorem}{Theorem}

\newtheorem{definition}{Definition}
\newtheorem{remark}{Remark}

\newtheorem{proposition}{Proposition}
\newtheorem{corollary}{Corollary}

\censoringfalse

\def\tsc#1{\csdef{#1}{\textsc{\lowercase{#1}}\xspace}}
\tsc{WGM}
\tsc{QE}

\ExplSyntaxOn
\cs_set:Npn \__first_footerline: {}
\ExplSyntaxOff
\begin{document}
\let\WriteBookmarks\relax
\def\floatpagepagefraction{1}
\def\textpagefraction{.001}

\shorttitle{A CBF-Constrained MPC Framework for Safe RL}    

\shortauthors{A. U. Kaypak, P. Krishnamurthy, F. Khorrami}  

\title [mode = title]{A Control Barrier Function-Constrained Model Predictive Control Framework for Safe Reinforcement Learning}  



%

\author[1]{Ali Umut Kaypak}[orcid=0009-0007-5040-6091]

\cormark[1]


\ead{ak10531@nyu.edu}


\credit{Conceptualization, Methodology, Software, Writing – original draft, Writing – review and editing}

\affiliation[1]{organization={Control/Robotics Research Laboratory, Electrical and Computer Engineering Department, Tandon School of Engineering, New York University},
            city={Brooklyn},
            postcode={11201}, 
            state={NY},
            country={USA}}

\author{Prashanth Krishnamurthy}[orcid=0000-0001-8264-7972]

\fnmark[a]

\ead{prashanth.krishnamurthy@nyu.edu}


\credit{Conceptualization, Methodology, Writing – review and editing, Supervision}

\author{Farshad Khorrami}[orcid=0000-0002-8418-004X]

\fnmark[a]
\ead{khorrami@nyu.edu}

\credit{Conceptualization, Methodology, Writing – review and editing, Supervision, Funding acquisition}

\cortext[1]{Corresponding author}

\begin{abstract}
Ensuring safety under unknown and stochastic dynamics remains a significant challenge in reinforcement learning (RL). In this paper, we propose a model predictive control (MPC)-based safe RL framework, called Probabilistic Ensembles with CBF-constrained Trajectory Sampling (PECTS), to address this challenge. PECTS jointly learns stochastic system dynamics with probabilistic neural networks (PNNs) and control barrier functions (CBFs) with Lipschitz-bounded neural networks. Safety is enforced by incorporating learned CBF constraints into the MPC formulation while accounting for the model stochasticity. This enables probabilistic safety under model uncertainty. To solve the resulting MPC problem, we utilize a sampling-based optimizer together with a safe trajectory sampling method that discards unsafe trajectories based on the learned system model and CBF. We also derive a finite-horizon exit-probability bound using tail information for centered barrier increments. We validate PECTS in various simulation studies, where it outperforms baseline methods. 
\end{abstract}


\maketitle

\section{Introduction}\label{sec:intro}
Safe RL addresses the problem of learning policies that satisfy safety constraints in safety-critical tasks to prevent irreversible damage to the robot or its environment. Prior work in this area encompasses a diverse range of approaches, including policy optimization-based, formal methods-based, control theory-based, and Gaussian process-based methods~\cite{safe_rl_review}. Among these, control theory-based approaches provide strong theoretical guarantees of safety. However, their applicability often depends on the availability of a system model or a predefined safety function. This dependency limits their practicality when environmental constraints and system dynamics are unknown or challenging to model accurately. Motivated by these limitations, we propose a model-based safe RL framework that jointly learns a stochastic system model and a CBF, ensuring probabilistic safety under model stochasticity.

We build PECTS upon two complementary research directions: CBFs for safety assurance~\cite{discrete_cbf, discrete_cbf_mpc, KAYPAK2026105263} and model-based reinforcement learning (MBRL)~\cite{pets_paper} for efficient policy learning under uncertainty. Discrete-time CBFs were initially utilized for deterministic systems in one-step safety filters~\cite{discrete_cbf} and MPC controllers~\cite{discrete_cbf_mpc}. More recently, their extensions to stochastic systems have been explored, and corresponding safety guarantees in such stochastic environments have been established~\cite{stochastic_discrete_cbf}. Separately, MBRL was shown to be effective in reducing the high sample complexity of model-free RL and in handling complex, stochastic dynamics through PNN–based system modeling~\cite{pets_paper, PDDM}. In this work, we unify these two research directions within a safe RL framework by combining PNN-based stochastic system modeling with CBF-based probabilistic safety guarantees for stochastic discrete-time systems.

On this foundation, we formulate PECTS as an MPC-based agent whose parameters are learned online. We incorporate CBF constraints within the MPC formulation and employ a sampling-based MPC optimizer to solve the resulting optimization problem. To enforce CBF constraints during optimization, we introduce a safe trajectory sampling mechanism that discards unsafe trajectories, ensuring the optimizer operates only on trajectories satisfying the CBF constraints. Specifically, our key contributions are as follows: 
\textbf{(1)} developing a novel CBF-based safe MBRL algorithm that explicitly accounts for stochasticity in the learned dynamics; \textbf{(2)} proposing a CBF learning methodology for stochastic discrete-time systems; \textbf{(3)} deriving a finite-horizon exit-probability bound that replaces an almost-sure bound on centered barrier increments with a tail bound; and  \textbf{(4)} demonstrating the effectiveness of the proposed safe MBRL method in extensive simulation studies.

\section{Related Work}\label{sec:related_works}
Standard RL algorithms are broadly categorized into model-based and model-free paradigms. This distinction extends to CBF-based safe RL. Notably, even when a CBF-based safe RL algorithm employs a model-free backbone for policy learning, it often relies on model information for safety enforcement. Consequently, the majority of CBF-based safe RL algorithms use model information~\cite{emam_safe_rl, airaldi2025probabilistically, sabouni_safe_rl, pmlr-v202-wang23as, NEURIPS2021_d71fa38b}. While purely model-free approaches exist~\cite{model_free_rl_cbf, model_free_rl_clbf_new}, they are often less sample-efficient than model-based approaches. Most model-based safe RL methods employing CBFs rely on varying degrees of prior knowledge about the system dynamics~\cite{emam_safe_rl, airaldi2025probabilistically, sabouni_safe_rl}. For instance, the approach presented in~\cite{emam_safe_rl} utilizes a known nominal system model and a known CBF, learning the unknown residual dynamics via Gaussian processes. In~\cite{airaldi2025probabilistically}, an MPC agent is defined and its parameters are optimized online, assuming a known system model. A shared limitation of these approaches is their dependence on a partially or fully known system model \textit{a priori}.

Many CBF-based safe RL methods assume a known CBF~\cite{emam_safe_rl, airaldi2025probabilistically, sabouni_safe_rl}. However, this assumption is restrictive for unknown environments. Moreover, even when the environmental constraints are known, constructing an analytical CBF remains a non-trivial task. To address this challenge, neural CBFs have been introduced for various scenarios~\cite{bolun_cbf_learning, robey_cbf_learning, bolun_2}. We integrate neural CBFs into an RL framework for unknown stochastic discrete-time systems, jointly learning both the system dynamics and the safety barrier. Similar to our approach, the methods in~\cite{pmlr-v202-wang23as, NEURIPS2021_d71fa38b}  jointly learn a barrier function and system dynamics, but they require safety regions specified \textit{a priori}. The approach in~\cite{song_safe_rl} circumvents this assumption, as we do, by learning a CBF from sensor data. However, their formulation restricts the CBF to an affine structure, limiting its expressiveness.

There are sampling-based MPC methods using CBFs for safety~\cite{9993190,11312129,YinJ-RSS-25}. For example, CBF constraints are used to define a trust region for action sampling in Model Predictive Path Integral (MPPI) in~\cite{9993190}. In~\cite{11312129}, MPPI is performed on a CBF-based guaranteed-safe system, ensuring that all sampled trajectories satisfy safety constraints. The approach in~\cite{YinJ-RSS-25} learns a neural CBF and incorporates its descent condition as a state constraint within a sampling-based MPC controller to enforce safety. Although PECTS shares the use of CBF-constrained sampling-based MPC with these methods, it is developed for safe RL settings where the system dynamics and safety constraints are initially unknown. This setting is beyond the scope of these approaches. 

\section{Problem Formulation}
\label{sec:problem_formulation}
Let $(\Omega, \mathcal{F}, \mathbb{P})$ be a probability space, and let $\mathcal{F}_0 \subset \mathcal{F}_1 \subset \dots \subset \mathcal{F}$ denote a filtration of $\mathcal{F}$. Consider a stochastic discrete-time system of the form
\begin{equation}
\label{eq:system_model}
    s_{k+1} = f(s_k, a_k, d_k), \quad \forall k \in \mathbb{Z}_{\geq0}
\end{equation}
where $f: \mathcal{S} \times \mathcal{A} \times \mathcal{W} \to \mathcal{S}$ is a continuous function, $s_k \in \mathcal{S} \subseteq \mathbb{R}^n$ and $a_k \in \mathcal{A} \subseteq \mathbb{R}^q$ denote the $\mathcal{F}_k$-measurable system state and control input (action), respectively, and $d_k \in \mathcal{W} \subseteq \mathbb{R}^d$ is an $\mathcal{F}_{k+1}$-measurable noise term. We assume that $\mathcal{S}$ is convex and that, for any $(s_k,a_k)\in\mathcal{S}\times\mathcal{A}$, the conditional distribution of $s_{k+1}$ induced by \eqref{eq:system_model} is unimodal. We also assume that all random variables and their functions are integrable in this paper.

For a safe set $\mathcal{C} \subset \mathcal{S} \subseteq \mathbb{R}^n$ and an initial state $s_0$ satisfying $s_0 \in \mathcal{C}$ almost surely (a.s.), the objective is to find a policy $\pi: \mathcal{S} \to \mathcal{A}$ that maximizes the expected finite-horizon cumulative reward while ensuring that the system remains within the safe set with probability at least $1- p$:
\begin{equation}
\label{eq:objective}
    \begin{aligned}
    \max_{\pi} \quad & J_r(\pi) = \mathbb{E}_{\pi}\!\left[ \sum_{k=0}^{M-1} r(s_k, a_k, d_k) \right] \\
     \text{s.t.} \quad & \mathbb{P}_{\pi}\!\left( s_k \in \mathcal{C},\ \forall k = 1, \dots, M \right) \ge 1-p
\end{aligned}
\end{equation}
where $M \in \mathbb{N}^{+}$ denotes the finite time horizon, and $r: \mathcal{S} \times \mathcal{A} \times \mathcal{W} \to \mathbb{R}$ is the reward function. The subscript $\pi$ denotes the trajectory distribution induced by the policy and the stochastic dynamics, including process disturbances and any randomness in the initial state. The functions $f$ and $r$, as well as the safe set $\mathcal{C}$, are assumed to be unknown to the agent. The agent observes the state $s_k$ at each time step and applies the input $a_k=\pi(s_k)$; after applying $a_k$, it observes the resulting reward $r(s_k,a_k,d_k)$ at the next time step. The agent is also equipped with a safety sensor that provides local safety information. At each time step $k$, the sensor returns a set of nearby states together with binary indicators specifying whether each such state is safe or unsafe.

\section{Methodology}
At each time step $k$, given the current system state $s_k$, PECTS plans over a finite horizon $H\leq M$ using a stochastic MPC objective. Specifically, it optimizes an input sequence $a_{0:H-1} := (a_0,\dots,a_{H-1})$ to maximize the expected cumulative reward under a learned stochastic model while enforcing a CBF condition:
\begin{align}
            &\!\!\!\!\!\!\!\!\max_{a_{0:H-1}} \; J(a_{0:H-1}) := \mathbb{E}\!\left[\sum_{\tau=0}^{H-1} \hat{r}_\tau \right] \label{eq:mpc_formulation}
        \\
        &\!\!\!\!\!\!\!\!e_\tau \sim \mathrm{Unif}\!\left(\{1,\dots,E\}\right), \quad \forall \tau =0, \dots, H-1, \notag
        \\ 
        &\!\!\!\!\!\!\!\!\hat{s}_{\tau+1} \sim f_{\theta,e_\tau}( \cdot \mid \hat{s}_\tau, a_\tau), \quad \forall \tau =0, \dots, H-1, \notag
        \\
        &\!\!\!\!\!\!\!\!\hat{r}_{\tau} \sim r_{\theta,e_\tau}(\cdot \mid \hat{s}_\tau, a_\tau), \quad \forall \tau =0, \dots, H-1, \notag \\
        &\!\!\!\!\!\!\!\!\mathbb{E}\!\left[ h_{\phi}(\hat{s}_{\tau+1}) \mid \hat{s}_{\tau}, a_{\tau}, e_\tau\right] \geq \kappa h_{\phi}(\hat{s}_\tau),  \: \, \forall \tau =0, \dots, H-1, \notag \\
        &\!\!\!\!\!\!\!\!a_\tau \in \mathcal{A}, \quad \forall \tau =0, \dots, H-1, \notag\\
        &\!\!\!\!\!\!\!\!\hat{s}_0 = s_{k}.\notag
\end{align}

In this formulation, $\hat{s}_\tau$ and $\hat{r}_\tau$ denote the rollout state and reward induced by an ensemble of $E$ PNNs with parameters $\theta$. The discrete index $e_\tau \in \{1,\dots,E\}$ is uniformly distributed over $E$ networks and selects the ensemble member used at rollout step $\tau$. The next-state and reward distributions are given by the corresponding PNN, denoted by $f_{\theta,e_\tau}(\cdot \mid \hat{s}_\tau,a_\tau)$ and $r_{\theta,e_\tau}(\cdot \mid \hat{s}_\tau,a_\tau)$, respectively. Each PNN parameterizes a one-step conditional Gaussian distribution with diagonal covariance, modeling aleatoric uncertainty. The conditional mean and covariance are nonlinear functions of the current state and action; this parameterization does not involve linearization of the system dynamics. The assumed unimodality of the true conditional distribution motivates the Gaussian modeling choice, without implying exact Gaussianity.  When the distribution family is known, a PNN that parameterizes that family can be used. The unimodality assumption concerns the shape of the true distribution, not uncertainty in estimating it from finite data. We use an ensemble of bootstrapped PNNs to approximate this epistemic uncertainty about the learned dynamics; see~\cite{pets_paper} for the details of PNNs.

The function $h_\phi$ in~(\ref{eq:mpc_formulation}) denotes the CBF, whose superlevel set, i.e., $\mathcal{C}_{\phi} := \{ s \in \mathcal{S} : h_{\phi}(s) \ge 0 \}$, represents a learned approximation of the safe set in~(\ref{eq:objective}). After solving~(\ref{eq:mpc_formulation}) via a sampling-based optimizer, the first system input in the solution sequence is applied to the system in a receding-horizon manner. The resulting state transition, reward, and safety-sensor labels are recorded in data buffers, which are used to periodically update the PNN ensemble and the CBF.

In subsection~\ref{sec:method_CBF}, we present the theoretical background of CBFs in stochastic discrete-time systems and describe how the CBF $h_{\phi}$ is learned within our algorithm. In subsection~\ref{sec:traj_desc}, we explain how the MPC problem in~\eqref{eq:mpc_formulation} is solved via a sampling-based optimizer.

\subsection{Safety Bounds and CBF Learning}
\label{sec:method_CBF}
For a given state-feedback controller $\pi:\mathcal{S} \rightarrow \mathcal{A}$, the closed-loop dynamics of the system (\ref{eq:system_model}) can be written as
\begin{equation}
\label{eq:closed_loop}
s_{k+1} = f(s_k, \pi(s_k), d_k)=:\tilde{f}(s_k,  d_k), \quad \forall k \in \mathbb{Z}_{\ge 0}.
\end{equation}
Let the safe set be defined by the superlevel set of a continuous function $h$, i.e., $\mathcal{C}=\{s\in\mathbb{R}^n \mid h(s)\ge 0\}$. For any $K \in \mathbb{N}^+$, the $K$-step exit probability is the probability that the stochastic closed-loop system~(\ref{eq:closed_loop}) leaves the safe set $\mathcal{C}$ within $K$ steps. Its formal definition is given as follows.
\begin{definition}[K-step Exit Probability \cite{stochastic_discrete_cbf}]
        Let $h : \mathbb{R}^n \to \mathbb{R}$ be a continuous function. 
For any $K \in \mathbb{N}^+$, and initial condition 
$s_0 \in \mathcal{C}$, the $K$-step exit probability of the closed-loop system (\ref{eq:closed_loop}) is given by:
\begin{equation}
P_u(K, s_0) 
= \mathbb{P} \left\{ \min_{k \in \{0, \ldots, K\}} h(s_k) < 0 \right\}. 
\end{equation}
    \end{definition}
The following theorem bounds the K-step exit probability when $h$ satisfies stochastic discrete-time CBF constraints.

\begin{theorem}[Upper-bound on K-step Exit Probability~\cite{stochastic_discrete_cbf}]
\label{thrm:1}
Let $K \in \mathbb{N}^+$, $\sigma > 0$, $\kappa \in [0, 1]$ and $\delta > 0$. If $\forall k \in \{0, \ldots, K-1\}$ the following inequalities hold a.s.:
\begin{align}
\mathbb{E}[h(s_{k+1}) \mid \mathcal{F}_k] - h(s_{k+1}) &\leq \delta, \label{eq:expected_ineq_thrm1} \\
\mathrm{Var}(h(s_{k+1}) \mid \mathcal{F}_k) &\leq \sigma^2, \label{eq:var_ineq_thrm1} 
\end{align} 
and the condition
\begin{equation}
\label{eq:stoch_cbf_cond}
\mathbb{E}\!\left[h\left(\tilde{f}(s_k, d_k)\right) \Big| \mathcal{F}_k\right] \geq \kappa h(s_k)
\end{equation}
is satisfied, then the $K$-step exit probability is bounded as
\begin{equation}
P_u(K, s_0) \leq T\left( \frac{\kappa^K h(s_0)}{\delta}, \frac{\sigma \sqrt{K}}{\delta} \right)
\end{equation}
where
$T(\lambda, \xi) := \left( \frac{\xi^2}{\lambda + \xi^2} \right)^{\lambda + \xi^2} e^{\lambda}.$
\end{theorem}
Theorem~\ref{thrm:1} bounds the probability of leaving the safe set over a prescribed finite horizon. This safety criterion is consistent with our finite-horizon safe RL formulation and motivates the stochastic barrier constraint used in the MPC design. In the following theorem, we provide an alternative upper bound that utilizes tail information for the centered barrier increments and can be informative in regimes where Theorem~\ref{thrm:1} is overly conservative.

\begin{theorem}[Upper-bound on K-step Exit Probability with a Tail Bound] \label{thrm:upper_bound_tail}
    Let $K \in \mathbb{N}^+$, $\sigma > 0$, $\kappa \in [0, 1]$ and $\delta > 0$. If $\forall k \in \{0, \ldots, K-1\}$ the following inequalities hold a.s.:
    \begin{align}
        \mathbb{P}\{\mathbb{E}[h(s_{k+1}) \mid \mathcal{F}_k] - h(s_{k+1}) > \delta \mid \mathcal{F}_k\} &\leq \zeta(\delta),\\
        \mathrm{Var}(h(s_{k+1}) \mid \mathcal{F}_k) &\leq \sigma^2, \label{eq:second_moment_ineq}
    \end{align}
    where $\zeta:(0, \infty)  \rightarrow [0,1]$ is a deterministic function and CBF condition~\eqref{eq:stoch_cbf_cond} is satisfied for all $k \in \{0, \dots, K-1\}$. Then, the K-step exit probability is bounded as
    \begin{equation}
      \label{eq:upper_bound_tail}
        P_u(K,s_0) \! \leq \min\!\left\{ \! 1, T\!\left(\! \frac{\kappa^K h(s_0)}{\delta}, \frac{\sigma \sqrt{S_K}}{\delta}  \right) + K \zeta(\delta) \!\right\},
    \end{equation}
    where
    \begin{equation}
        S_K := \begin{cases}
            \frac{1 - \kappa^{2K}}{1 - \kappa^2},\quad 0 \leq \kappa <1 \\
            K, \quad \kappa=1,
        \end{cases}
    \end{equation}
    and $T$ is defined in Theorem~\ref{thrm:1}.
\end{theorem}

\begin{proof}
    The proof follows the martingale construction through Doob's decomposition as in \cite{stochastic_discrete_cbf} and then applies Corollary 2.3 in \cite{hoeffding_ineq} to accommodate increments exceeding the cutoff. For $\kappa =0$, the asserted bound equals $1$ and the result is trivial. Thus, suppose $\kappa \in (0,1]$. Construct the supermartingale candidate $W_k$ as follows:
    \begin{equation}
        W_k := \kappa^K h(s_0) - \kappa^{K-k}h(s_k),
        \qquad k=0,\ldots,K.
    \end{equation}
    Note that $W_0=0$, and by using the stochastic CBF condition \eqref{eq:stoch_cbf_cond} at time $k-1$, we obtain, for $k=1,\ldots,K$,
    \begin{align}
        \!\!\!\!\!\!\!\!\!\!\!\!\mathbb{E}[W_k \mid \mathcal{F}_{k-1}]
        &= \kappa^K h(s_0)
        - \kappa^{K-k}\mathbb{E}[h(s_k)\mid \mathcal{F}_{k-1}]  \notag \\
        &\leq \kappa^K h(s_0)
        \!-\! \kappa^{K-(k-1)}h(s_{k-1}) \!=\! W_{k-1},
    \end{align}
    showing $W_k$ is a supermartingale. We obtain the martingale from Doob's decomposition as
    \begin{equation}
        M_k := W_k + \sum_{i=1}^k
        \left(W_{i-1} - \mathbb{E}[W_i \mid \mathcal{F}_{i-1}]\right)
        \geq W_k.
    \end{equation}
    Define the martingale increments, for $k=1,\ldots,K$, as
    \begin{equation}
        \xi_k := M_k - M_{k-1}
        = \kappa^{K-k}
        \left(\mathbb{E}[h(s_k)\mid \mathcal{F}_{k-1}] - h(s_k)\right).
    \end{equation}
    Hence, $M_k = \sum_{i=1}^k \xi_i$ and
    $\mathbb{E}[\xi_k \mid \mathcal{F}_{k-1}] = 0$, making the increments satisfy the supermartingale-difference requirement of Corollary 2.3 in~\cite{hoeffding_ineq}.
    Now, we apply the truncation inequality
    \cite[Eq.~(24)]{hoeffding_ineq}, which follows from
    Corollary~2.3, with $n=K$, $X_k=M_k$,
    $x=\kappa^Kh(s_0)$, $y=\delta$, and
    $v=\sigma\sqrt{S_K}$.
    By the conditional-variance assumption applied at time $i-1$,
    \begin{equation}
    \label{eq:variance_est}
    \begin{split}
        \langle M\rangle_K
        &:= \sum_{i=1}^K
        \mathbb{E}[\xi_i^2\mid\mathcal{F}_{i-1}] \\
        &= \sum_{i=1}^K \kappa^{2(K-i)}
        \operatorname{Var}(h(s_i)\mid\mathcal{F}_{i-1})
        \leq \sigma^2S_K=v^2
    \end{split}
    \end{equation}
    a.s. Hence, the additional term
    $\mathbb{P}\{\langle M\rangle_K>v^2\}$ in the
    cited inequality is zero.
    Since $M_k\geq W_k$, an exit $h(s_k)<0$ implies
    $M_k>\kappa^Kh(s_0)$. The cited inequality therefore gives
    \begin{align}
        P_u(K,s_0)
        \leq
        \mathbb{P}\!\left\{
            \max_{1\leq k\leq K}M_k\geq\kappa^Kh(s_0)
        \right\}\nonumber\\
        \leq T\!\left(
            \frac{\kappa^Kh(s_0)}{\delta},
            \frac{\sigma\sqrt{S_K}}{\delta}
        \right)
        +\sum_{i=1}^K\mathbb{P}\{\xi_i>\delta\}.
    \end{align}
    Since $0<\kappa^{K-i}\leq1$, the assumed tail bound,
    applied at time $i-1$ for $i=1,\ldots,K$, implies
    $\mathbb{P}\{\xi_i>\delta\mid\mathcal{F}_{i-1}\}
    \leq\zeta(\delta)$.
    Taking expectations and summing gives $\sum_{i=1}^K\mathbb{P}\{\xi_i>\delta\}
        \leq K\zeta(\delta).$ Combining these bounds with $P_u(K,s_0)\leq1$ gives \eqref{eq:upper_bound_tail}, completing the proof.
\end{proof}
CBF deviations have sub-Weibull upper tails if there exist deterministic constants $A\geq1$, $b>0$, and $\alpha>0$ such that, for every $k\in\{0,\ldots,K-1\}$ and $\delta>0$,
\begin{equation}
\label{eq:tail_ineq}
    \mathbb{P}\!\left\{
        \mathbb{E}[h(s_{k+1})\mid\mathcal{F}_k]
        -h(s_{k+1})>\delta
        \,\middle|\,\mathcal{F}_k
    \right\}
   \! \leq\! Ae^{\left[-\left(\frac{\delta}{b}\right)^\alpha\right]}
\end{equation}
a.s. In this case, $\zeta(\delta)=\min\{1,Ae^{[-(\delta/b)^\alpha]}\}$ can be chosen in~\eqref{eq:upper_bound_tail}. The cases $\alpha=2$ and $\alpha=1$ correspond to Gaussian and exponential tail decay, respectively. 

In the following proposition, we provide a bound on the conditional variance required by Theorems~\ref{thrm:1} and \ref{thrm:upper_bound_tail}.
\begin{proposition}[Conditional Variance Bound]
\label{prop:1}
Suppose \(h:\mathcal{S}\to \mathbb{R}\) is globally Lipschitz with constant \(L_h \ge 0\) and bounded. Also assume that \(\tilde{f}\) is globally Lipschitz in the noise argument \(d_k\) with constant \(L_f \ge 0\), i.e.,
\(
\| \tilde{f}(s_k, \hat d_k) - \tilde{f}(s_k, \bar d_k) \|
\le L_f \|\hat d_k - \bar d_k\|
\quad \forall s_k \in \mathcal{S}, \ \forall \hat d_k, \bar d_k \in \mathbb{R}^d.
\)
Let the noise \((d_k)_{k \ge 0}\) have uniformly bounded conditional covariance, i.e., there exists a constant \(\sigma_d^2\) such that \(\operatorname{tr}(\operatorname{Cov}(d_k \mid \mathcal{F}_k)) \le \sigma_d^2\) for all \(k\) a.s. Then,
\begin{equation}
\mathrm{Var}(h(s_{k+1})\mid \mathcal{F}_k) \le L_h^2 L_f^2 \sigma_d^2 .\label{eq:ineq_variance}
\end{equation}
\end{proposition}

\begin{proof}
 We use the conditional-variance bound: for a $Y \in \mathcal{L}^2$ and any $\mathcal{F}_k$-measurable random variable $c \in \mathcal{L}^2$, $\mathrm{Var}(Y \mid \mathcal{F}_k) \le \mathbb{E}[(Y-c)^2 \mid \mathcal{F}_k]$. Let $Y = h(s_{k+1})$ and let $c = h(\tilde{f}(s_k, \mathbb{E}[d_k \mid \mathcal{F}_k]))$. Note that both random variables belong to $\mathcal{L}^2$ since $h$ is assumed to be bounded, and $c$ is $\mathcal{F}_k$-measurable. Then:
\begin{align}
    &\mathrm{Var}(h(s_{k+1}) \mid \mathcal{F}_k)\\ 
    &\le \mathbb{E}\!\left[\left(h(s_{k+1}) - h\!\left(\tilde{f}(s_k, \mathbb{E}[d_k\mid \mathcal{F}_k])\right)\right)^2 \Bigm| \mathcal{F}_k \right] \label{eq:var_ineq} \\
    &\le L_h^2\, \mathbb{E}\!\left[\left\| s_{k+1} - \tilde{f}(s_k, \mathbb{E}[d_k\mid \mathcal{F}_k]) \right\|_2^2 \Bigm| \mathcal{F}_k \right] \label{eq:prf_liph} \\
    &= L_h^2\, \mathbb{E}\!\left[\left\| \tilde{f}(s_k, d_k)  - \! \tilde{f}(s_k, \mathbb{E}[d_k\mid \mathcal{F}_k]) \right\|_2^2 \Bigm| \mathcal{F}_k \right] \\
    &\le L_h^2 L_f^2 \, \mathbb{E}\!\left[ \left\| d_k - \mathbb{E}[d_k\mid \mathcal{F}_k] \right\|_2^2 \Bigm| \mathcal{F}_k \right] \label{eq:prf_lipf} \\
    &= L_h^2 L_f^2\, \operatorname{tr}\!\left(\operatorname{Cov}(d_k\mid \mathcal{F}_k)\right) \le L_h^2 L_f^2 \sigma_d^2.\label{eq:trace_def} 
\end{align}
Inequalities \eqref{eq:prf_liph} and \eqref{eq:prf_lipf} follow from the Lipschitz properties of $h$ and $\tilde{f}$ respectively. Equation \eqref{eq:trace_def} follows from the definition of the trace of the covariance matrix for a random vector. Thus, inequality~\eqref{eq:ineq_variance} holds.
\end{proof}

To ensure that the variance bound in Proposition~\ref{prop:1} remains meaningful and to prevent the Lipschitz constant $L_h$ of the learned CBF $h_\phi$ from being excessively large during training, we employ the Lipschitz-bounded networks proposed in \cite{lipschitz_nn}. Consequently, the training process adheres to the user-defined upper bound on the Lipschitz constant of $h_\phi$. Furthermore, to ensure that the output of $h_\phi$ remains bounded, we utilize a $\tanh$ activation function, which is 1-Lipschitz, in the final layer. The bounded output satisfies the assumption \eqref{eq:expected_ineq_thrm1} in Theorem~\ref{thrm:1} with $\delta =2$. For bounded disturbances, a smaller $\delta$ can be obtained as shown in Proposition~4 in the extended version of \cite{stochastic_discrete_cbf}.

The uniform choice $\delta =2$ uses only the output range of $h_\phi$. Theorem~\ref{thrm:upper_bound_tail} allows the exit-probability bound to account for the disturbance tail through $\zeta(\delta)$. For Gaussian disturbances, the Lipschitz assumptions above provide a sub-Weibull upper-tail bound. The following corollary combines this tail bound with the variance bound in Proposition~\ref{prop:1} to obtain an explicit Gaussian specialization of Theorem~\ref{thrm:upper_bound_tail}. 

\begin{corollary}[Upper Bound on K-step Exit Probability under Gaussian Disturbances] \label{corollary:1}
    Let $K\in\mathbb{N}^+$ and $\kappa \in [0,1]$. Suppose the assumptions of Proposition~\ref{prop:1} hold with $L_h, L_f, \sigma_d >0$. In addition, assume that for every $k \in \{0, \dots , K-1 \}$, the stochastic CBF condition~\eqref{eq:stoch_cbf_cond} holds and the disturbance is conditionally Gaussian, i.e.,
    \begin{equation}
        d_k | \mathcal{F}_k \sim \mathcal{N}(\mu_{d,k}, \Sigma_{d,k}) \quad \text{a.s.}
    \end{equation}
    Then, for every $\delta > 0$, $ P_u(K,s_0)$ is bounded as
    \begin{equation} P_u(K,s_0) \leq
        \min\left\{1, T\left(\frac{\kappa^K h(s_0)}{\delta}, \frac{\sigma \sqrt{S_K}}{\delta}\right)\! + \! Ke^{\left(-\frac{\delta^2}{2\sigma^2}\right)} \! \right\},
    \end{equation}
    where $T$ and $S_K$ are defined as in Theorem \ref{thrm:1} and Theorem \ref{thrm:upper_bound_tail}, respectively, and $\sigma = L_hL_f\sigma_d$.
\end{corollary}
\begin{proof}
    Proposition~\ref{prop:1} shows that \eqref{eq:second_moment_ineq} is satisfied with $\sigma = L_hL_f\sigma_d$. Fix $k \in \{0, \dots, K-1\}$ and consider the conditional distribution given $\mathcal{F}_k$. Let $Z\sim\mathcal{N}(0,I_d)$ be an auxiliary Gaussian vector, independent of $\mathcal{F}_k$, with independent standard normal components. Define $g_k(z) := h\left( \tilde{f}\left( s_k, \mu_{d,k} + \Sigma^{1/2}_{d,k}z \right) \right).$ By the Gaussian assumption, $h(s_{k+1})$ has the same conditional distribution as $g_k(Z)$. The Lipschitz constant of $g_k$ is at most
    \begin{equation}
        L_hL_f\|\Sigma^{1/2}_{d,k}\|_{op} \leq L_hL_f \sqrt{\operatorname{tr}(\Sigma_{d,k})} \leq \sigma = L_hL_f\sigma_d.
    \end{equation}
    Thus, the same $\sigma$ used in the variance bound from Proposition \ref{prop:1} also bounds the Lipschitz constant required by Gaussian concentration. Applying the Gaussian concentration inequality~\cite[Theorem~5.6]{concentration_ineq} to $-g_k$ gives, for every $\delta >0$,
    \begin{equation}
        \mathbb{P}\!\left\{
            \mathbb{E}[h(s_{k+1}) \mid \mathcal{F}_k]
            - h(s_{k+1}) > \delta
            \,\middle|\, \mathcal{F}_k
        \right\}
        \leq e^{-\frac{\delta^2}{2\sigma^2}}
    \end{equation}
    a.s. Applying Theorem~\ref{thrm:upper_bound_tail} with this variance bound and $\zeta(\delta) = e^{-\frac{\delta^2}{2\sigma^2}}$ completes the proof.
\end{proof}
Exit-probability bounds analogous to that in Corollary~\ref{corollary:1} can be derived for other noise distributions whenever suitable concentration inequalities and uniform conditional variance bounds hold for the centered barrier increments. Note that such extensions are not required for the validity of our theoretical framework, as Theorem~\ref{thrm:1} already provides an exit-probability bound with an almost-sure bound on the centered barrier increments and our bounded CBF satisfies it. Theorem~\ref{thrm:upper_bound_tail} replaces this almost-sure boundedness assumption with a tail bound and can provide informative finite-horizon bounds in regimes where the original theorem is overly conservative. Corollary~\ref{corollary:1} shows this general result for Gaussian disturbances, although exploiting distribution-specific properties more directly can give tighter bounds.

For the learned CBF $h_\phi$, the preceding exit-probability bounds apply to the physical safe set whenever $\mathcal{C}_\phi \subseteq \mathcal{C}$ and $s_0 \in \mathcal{C}_\phi$. The latter condition can be checked at the start of each episode by evaluating \(h_\phi(s_0)\) and verifying its nonnegativity. Constructing $h_{\phi}$ with Lipschitz-bounded networks also helps circumvent the intractability of computing $\mathbb{E}[ h_{\phi}(\hat{s}_{\tau+1})| \allowbreak \hat{s}_{\tau}, a_\tau, \allowbreak e_\tau ]$ in (\ref{eq:mpc_formulation}). In general, this term does not admit a closed-form expression due to the nonlinearity of the neural network $h_{\phi}$, and its estimation would require costly sampling within the MPC loop. The following proposition enables the derivation of a tractable lower bound on this term and facilitates the construction of a surrogate safety constraint for the CBF condition in (\ref{eq:mpc_formulation}):
\begin{proposition}[Conditional Expectation Lower Bound]
\label{prop:2}
    Suppose \(h:\mathcal{S}\to \mathbb{R}\) is globally Lipschitz with constant \(L_h \ge 0\). Then:
    \begin{equation}
    \label{eq:lower_bound}
    \begin{aligned}
        &\mathbb{E}\!\left[ h(s_{k+1} ) \mid \mathcal{F}_k \right] \geq h(\mathbb{E}\!\left[ s_{k+1} \mid \mathcal{F}_k\right]) \\ &\quad \quad \quad \quad \quad \quad \quad \quad \quad - L_h\sqrt{\operatorname{tr}\!\left(\operatorname{Cov}\left(s_{k+1}\mid \mathcal{F}_k\right) \right)}.
            \end{aligned}
    \end{equation}
\end{proposition}
\begin{proof}
        \begin{align}
            &|\mathbb{E}\!\left[ h(s_{k+1})\mid \mathcal{F}_k\right] - h(\mathbb{E}\!\left[s_{k+1}\mid \mathcal{F}_k\right])| \label{eq:prop2_1}\\
            & \quad \leq L_h \mathbb{E}\!\left[ \|s_{k+1} - \mathbb{E}\!\left[ s_{k+1} \mid \mathcal{F}_k\right] \|_2 \mid \mathcal{F}_k \right] \label{eq:prop2_2} \\
            & \quad \leq L_h \sqrt{\mathbb{E}\!\left[ \| s_{k+1} - \mathbb{E}\!\left[ s_{k+1} \mid \mathcal{F}_k\right] \|_2^2 \mid \mathcal{F}_k\right]} \label{eq:prop2_3} \\
            & \quad = L_h\sqrt{\operatorname{tr}\!\left(\operatorname{Cov}\left(s_{k+1}\mid \mathcal{F}_k\right) \right)}. \label{eq:prop2_4}
        \end{align}
    Inequalities (\ref{eq:prop2_2}) and (\ref{eq:prop2_3}) follow from the Lipschitz property of $h$ and the nonnegativity of the variance, respectively. Hence, inequality (\ref{eq:lower_bound}) holds.     
\end{proof}
Let $\mu_e(s, a)$ and $\Sigma_e(s,a)$ denote the mean and covariance matrix of the next state predicted by the $e$-th PNN in an ensemble of size $E$. Define the surrogate CBF margin
\begin{equation}
    \Gamma_e(s,a):= h_\phi(\mu_e(s,a)) - \kappa h_\phi(s) - L_h\sqrt{\operatorname{tr}(\Sigma_e(s,a))}. 
\end{equation}
By Proposition~\ref{prop:2}, nonnegativity of this margin is sufficient
for the original CBF condition in~\eqref{eq:mpc_formulation}. We therefore impose the following constraint for all $\tau=0,\dots,H-1$:
 \begin{equation}
  \Gamma_{e_{\tau}}(\hat{s}_\tau,a_\tau) \geq 0.  \label{eq:surrogate_safety}
 \end{equation}
Although this constraint shrinks the original MPC feasibility set in ($\ref{eq:mpc_formulation}$), it avoids evaluating the intractable expectation.

The CBF $h_\phi$ is trained within our framework as follows. Let $\mathcal{D}^{+}$ and $\mathcal{D}^{-}$ denote the safe and unsafe data buffers, respectively, populated via the safety sensor outputs. Let $\mathcal{D}^{\mathrm{fea}}$ represent the dataset of visited state-action pairs. We train the CBF $h_\phi$ by minimizing the composite loss function:
\vspace*{-5mm}
\begin{equation}
\label{eq:cbf_loss}
    \mathcal{L}_{h_\phi} = \lambda_{1} \mathcal{L}_+ + \lambda_2 \mathcal{L}_- + \lambda_3\mathcal{L}_{\mathrm{fea}}.
\end{equation}
The terms $\mathcal{L}_{+}$ and $\mathcal{L}_{-}$ penalize violations of $h_\phi(s) \geq \epsilon_+$ on safe samples and $h_\phi(s) \leq - \epsilon_-$ on unsafe samples:
\vspace*{-1mm}
\begin{equation}
\label{eq:epsilons}
\begin{aligned}
    \mathcal{L}_{+} &=  \frac{1}{|\mathcal{D}^{+}|}\sum_{s \in \mathcal{D}^{+}}  \left[-h_{\phi}\left(s \right) + \epsilon_{+}\right]^+, \\
    \mathcal{L}_{-} &=  \frac{1}{|\mathcal{D}^{-}|}\sum_{s \in \mathcal{D}^{-}} \left[ h_{\phi}\left(s \right) + \epsilon_{-}\right]^+
\end{aligned}
\end{equation}
where $[x]^+ := \max(0,x)$. Finally, $\mathcal{L}_{\mathrm{fea}}$ promotes the CBF feasibility condition (\ref{eq:surrogate_safety}) using the learned system model: 
\begin{equation}
\mathcal{L}_{\mathrm{fea}} =\frac{1}{E|\mathcal{D}^{\text{fea}}|}\sum_{e=1}^E\sum_{(s,a) \in \mathcal{D}^{\mathrm{fea}}} [-\Gamma_e(s,a) +\epsilon_{\mathrm{fea}}  ]^+ \label{eq:fea_loss}
\end{equation}
Next, we relate the learned-model constraint \eqref{eq:surrogate_safety} to the true-system CBF condition \eqref{eq:stoch_cbf_cond}. 
\begin{proposition}[Model-error-dependent closed-loop
barrier condition]
\label{prop:model_error_cbf} Consider a fixed learned barrier $h_\phi$  and dynamics ensemble. Suppose $h_\phi$ is bounded and globally Lipschitz with constant $L_h\geq 0$, and the true next state has finite conditional second moments. Let $\kappa\in[0,1]$. For an executed action $a_k$, let $e$ be an ensemble member used to verify its first-stage constraint. Suppose there exist deterministic constants $\varepsilon_\mu,\varepsilon_q\geq 0$ such that, a.s. at each step,
\begin{align}
&\left\|
\mathbb{E}[s_{k+1}\mid\mathcal{F}_k]
-\mu_e(s_k,a_k)
\right\|_2
\leq \varepsilon_\mu, \label{eq:conditional_first_moment_err}\\
&\sqrt{\operatorname{tr}
\bigl(\operatorname{Cov}(s_{k+1}\mid\mathcal{F}_k)\bigr)} -
\sqrt{\operatorname{tr}
\bigl(\Sigma_e(s_k,a_k)\bigr)}
\leq \varepsilon_q.\label{eq:conditional_second_moment_err}
\end{align}
Here, all conditional moments are taken under the true dynamics, with $a_k$ and $e$ measurable with respect to $\mathcal{F}_k$. Suppose the MPC prediction is initialized at the measured state, $\hat{s}_0=s_k$, and the executed action satisfies the first-stage constraint in \eqref{eq:surrogate_safety}. Define its surrogate slack as $\nu_k:= \Gamma_e(s_k,a_k) \geq 0$. Then, with $r_{\mathrm{mod}} :=L_h(\varepsilon_\mu+\varepsilon_q)$, the true closed-loop transition satisfies
\begin{equation}
\label{eq:model_error_cbf}
\mathbb{E}[h_\phi(s_{k+1})\mid\mathcal{F}_k]
\geq
\kappa h_\phi(s_k)+\nu_k-r_{\mathrm{mod}}
\end{equation}
a.s.  Whenever $\nu_k\geq r_{\mathrm{mod}}$, the
true-system stochastic CBF condition \eqref{eq:stoch_cbf_cond} is recovered.
\end{proposition}
\begin{proof}
Proposition~\ref{prop:2}, Lipschitz continuity, and the moment-error bounds \eqref{eq:conditional_first_moment_err}–\eqref{eq:conditional_second_moment_err} give
\begin{align}
&\!\!\!\!\!\!\mathbb{E}[h_\phi(s_{k+1})\mid\mathcal{F}_k] \notag \\
&\!\!\!\!\!\!\geq h_\phi\bigl(\mu_e(s_k,a_k)\bigr)
    - L_h\left\|
        \mathbb{E}[s_{k+1}\mid\mathcal{F}_k]
        - \mu_e(s_k,a_k)
    \right\|_2 \notag \\
&\!\!\!\!\!\!\quad - L_h\sqrt{
    \operatorname{tr}\bigl(
        \operatorname{Cov}(s_{k+1}\mid\mathcal{F}_k)
    \bigr)} \notag \\
&\!\!\!\!\!\!\geq h_\phi\bigl(\mu_e(s_k,a_k)\bigr)
    - L_h\sqrt{\operatorname{tr}\bigl(\Sigma_e(s_k,a_k)\bigr)} - L_h(\varepsilon_\mu+\varepsilon_q) \notag \\
&\!\!\!\!\!\!= \kappa h_\phi(s_k)+\nu_k-r_{\mathrm{mod}}.
\end{align}
The final claim follows when $\nu_k\geq r_{\mathrm{mod}}$.
\end{proof}

Proposition~\ref{prop:model_error_cbf} shows that the true-system stochastic CBF condition is satisfied when the surrogate slack compensates for the dynamics-model error. For a fixed CBF, if this slack condition holds a.s. at every step of the analyzed horizon, Theorems~\ref{thrm:1} and~\ref{thrm:upper_bound_tail} apply under their remaining assumptions. Our feasibility loss~\eqref{eq:fea_loss} encourages the surrogate margin to be at least $\epsilon_{\mathrm{fea}}$ on the training dataset, promoting a margin that can absorb model error. Moreover, using Lipschitz-bounded networks to parameterize $h_\phi$ provides a uniform upper bound $L_h$ throughout training. This controls the amplification of moment errors by the learned CBF.

\subsection{Learning System Dynamics and Safe Trajectory Sampling}
\label{sec:traj_desc}
We use an ensemble of bootstrapped PNNs, following the method in~\cite{pets_paper}, to learn both the system dynamics and the reward function in~(\ref{eq:mpc_formulation}). Each PNN in the ensemble outputs a Gaussian distribution over the next state and reward, conditioned on the current state and system input.

Our sampling-based MPC solver works as follows. At each control step, it samples filtered input sequences around a shifted previous solution and evaluates them through particle rollouts with CBF checks. Input prefixes violating the CBF constraints are replaced with the safe prefixes from the current batch. Predicted returns are then used to refine the sampling mean, and a second checked batch is used to select the highest-return sequence, whose first action is executed. Repeated search failures trigger recovery. Figure~\ref{fig:optimizer_flow} summarizes this optimization flow.

\begin{figure}
    \centering
\includegraphics[width=0.92\columnwidth]{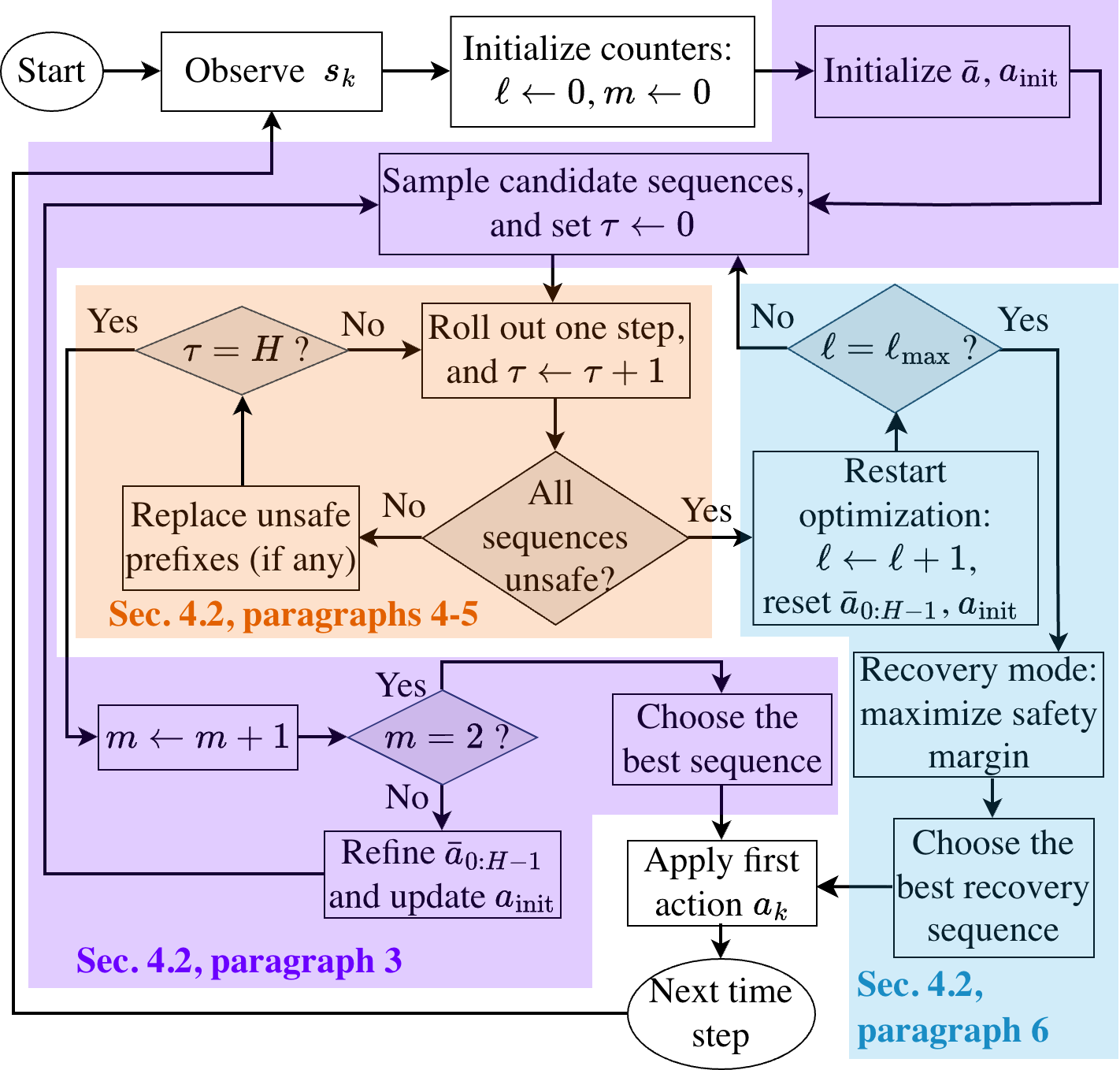}
    \caption{Optimization flow in PECTS. Unsafe input prefixes are replaced by safe ones; after $\ell_{\max}=5$ failed attempts, recovery is activated. Here, $\ell$ is the reoptimization-attempt counter, and $m$ is the optimization-stage counter. Color coding: purple for nominal optimization, orange for unsafety handling, and blue for recovery mode.}
    \label{fig:optimizer_flow}
    \vspace*{-2mm}
\end{figure}

 We use the input filtering technique utilized in~\cite{pineda2021mbrl} to generate smooth candidate input sequences. Let $a^{\star}_{0:H-1}$ denote the MPC solution from the previous time step ($a^{\star}_{0:H-1}\allowbreak=0$ for the first MPC call). We initialize the mean of the input sequence sampling distribution $\bar{a}_{0:H-1}$ as $\bar a_{\tau}=a^{\star}_{\tau+1}$ for $\tau=0,\dots,H-2$ and $\bar a_{H-1}= a^{\star}_{H-1}$, and set $a_{\mathrm{init}}=\bar a_{0}$. Let $N$ be the batch size, $\beta \in [0,1]$ the filtering coefficient, and $\Sigma_a$ the action-noise covariance. We generate the input sequences $\forall i \! \in \! \{0, \dots, N-1\},\: \tau \! \in \! \{0,\dots, H-1\}$ as follows:
 \vspace*{-5mm}
\begin{equation}
\label{eq:action_gen}
    \begin{aligned}
    n^i_\tau &\sim \mathcal{N}(\bar{a}_\tau, \Sigma_a),\\
    a^i_\tau &= \beta n^i_\tau + (1 - \beta) a^i_{\tau-1}, \, \text{where} \quad a^i_{-1}= a_{\mathrm{init}}.
    \end{aligned}
\end{equation}
We truncate Gaussian perturbations componentwise before filtering and clip the resulting actions to the action bounds before rollout. We do not use the standard MPPI weighted average as the final control output because a weighted average of safe input sequences is not necessarily safe due to the potential non-convexity of the safe set. Instead, we use the MPPI rule to refine the control distribution using the estimated cumulative rewards $\hat{R}_i$. $\hat{R}_i$ is computed as the average of the predicted cumulative rewards over all particles propagated under input sequence $a_{0:H-1}^i$ during the MPC rollout. The mean input sequence is updated as
\begin{equation}
    \bar{a}_\tau = \frac{\sum_{i=0}^{N-1} e^{\eta \hat{R}_i}a_\tau^i}{\sum_{i=0}^{N-1}e^{\eta \hat{R}_i}} \quad \forall \tau \in \{0, \dots, H-1\}
\end{equation}
where $\eta$ is the reward-scaling parameter. We then sample a new batch of input sequences around the updated mean $\bar{a}_{0:H-1}$ using (\ref{eq:action_gen}), setting $a_{\mathrm{init}} =\bar{a}_0$ prior to resampling, and select the single sequence yielding the maximum predicted cumulative reward as the optimizer's output.

Our safe trajectory sampling method for a given system input sequence $a_{0:H-1}^i$ and an initial state $\hat{s}_0$ proceeds as follows. We first generate identical initial states $\hat{s}_{0}^{i,j}$ from $\hat{s}_0$, for each particle $j = 0, \dots, P-1$. We then propagate each particle utilizing the PNNs from the model ensemble, i.e., $\hat{s}_{\tau+1}^{i,j} \!\sim \! \mathcal{N}(\mu_{e_\tau^{i,j}}(\hat{s}_{\tau}^{i,j},a_{\tau}^i), \Sigma_{e_\tau^{i,j}}(\hat{s}_{\tau}^{i,j},a_{\tau}^i))$. The bootstrap index $e_\tau^{i,j}$ is selected uniformly and independently for each particle at each rollout step. At each rollout step, we check whether each particle satisfies the surrogate safety condition in~(\ref{eq:surrogate_safety}) using the PNN assigned to that particle.

For each particle $j$, we verify whether the safety constraint $\Gamma_{e_{\tau}^{i,j}}( \hat{s}_{\tau}^{i,j}, a_\tau^i ) \geq 0$ is satisfied for all $\tau = 0, \dots, H-1$. If this condition is violated for any particle at any rollout step $\tau$, the corresponding input sequence $a_{0:H-1}^i$ is identified as unsafe. During the optimization process, when an input sequence is deemed unsafe at a rollout step $\tau < H$, the inputs up to $\tau$, i.e., $a_{0:\tau}$, are replaced with a randomly selected safe input prefix from the input sequence batch.  The propagated particle states and the cumulative reward are then updated up to the rollout step $\tau$ to remain consistent with the modified input sequence. The rollout subsequently continues from step $\tau+1$ using the remaining inputs. The remaining inputs beyond the violation step are not replaced, in order to preserve the diversity of the input sequence batch. If another safety violation occurs at a later rollout step, the same replacement procedure is applied. The replacement of unsafe input sequence prefixes with safe prefixes is analogous to the rewiring step in Resampling-Based Rollouts~\cite{YinJ-RSS-25}.

If all input sequences in the batch are identified as unsafe at rollout step $\tau$, we reset $\bar{a}_{0:H-1} =0$ and $a_{\mathrm{init}}=0$, and restart optimization. We repeat this procedure for a fixed number of attempts (set to $5$ in our implementation). Once this limit is reached, the algorithm enters a recovery mode because no input sequence satisfying the hard CBF constraints has been found. In this mode, the following MPC problem is solved without enforcing CBF constraints:
\begin{align}
        &\max_{a_{0:H-1}} \; \mathbb{E}\left[\sum_{\tau = 0}^{H-1} \frac{\Gamma_{e_\tau}(\hat{s}_\tau, a_\tau) }{\tau + 1}  \right] \notag
        \\
        &e_\tau \sim \mathrm{Unif}\!\left(\{1,\dots,E\}\right), \quad \forall \tau =0, \dots, H-1, \label{eq:soft_mpc_formulation}\\
        &\hat{s}_{\tau+1} \sim f_{\theta, e_\tau}( \cdot | \hat{s}_\tau, a_\tau), \quad \forall \tau =0, \dots, H-1,
        \notag \\
         &a_\tau \in \mathcal{A}, \quad \forall \tau =0, \dots, H-1, \notag\\
         &\hat{s}_0 = s_{k}.\notag
         \end{align}
The recovery mode ignores task reward and maximizes a safety-margin objective to regain feasibility.

\begin{remark}[Multistep predictive distributions]
    The Gaussian parameterization of each PNN applies to its one-step conditional predictions, not to the multistep state marginals conditioned on the current state and a planned action sequence. These marginals need not be Gaussian or unimodal. We approximate their evolution through particle-based Monte Carlo propagation using the nonlinear ensemble models, without projecting the propagated state distribution onto a single Gaussian at each step. Consequently, the particle representation can accommodate multistep non-Gaussian and multimodal predictions.
\end{remark}

\vspace*{-5mm}
\section{Experiments}
\label{sec:experiments}
We evaluated PECTS on two different tasks using three system models: a unicycle model, an Ackermann car model, and a 2D double integrator. Across all experiments, we set $\kappa=0.95$ and $L_h = 1$ in~(\ref{eq:surrogate_safety}). We chose an MPC horizon of $H=25$ for the unicycle and 2D double integrator, and $H=40$ for the Ackermann car.  All models were discretized using Euler's method with $dt=0.02$, and Gaussian noise was added to model system uncertainty. The specific dynamics for each model are described below.

\begin{figure}
    \centering
      \begin{subfigure}{0.38\textwidth}
        \centering
        \includegraphics[width=\linewidth]{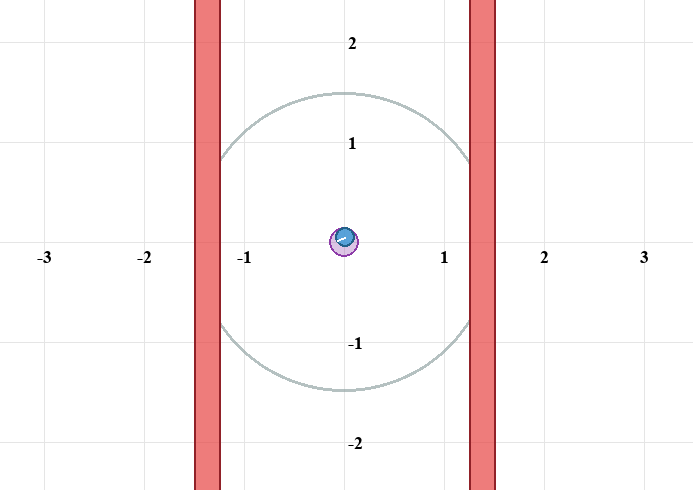}
        \caption{Circular path-following task. The gray circle (radius 1.5\,m) denotes the target circular path that yields the highest reward.}
        \label{fig:sub1}
      \end{subfigure}%
      \hfill
      \begin{subfigure}{0.38\textwidth}
        \centering
        \includegraphics[width=\linewidth]{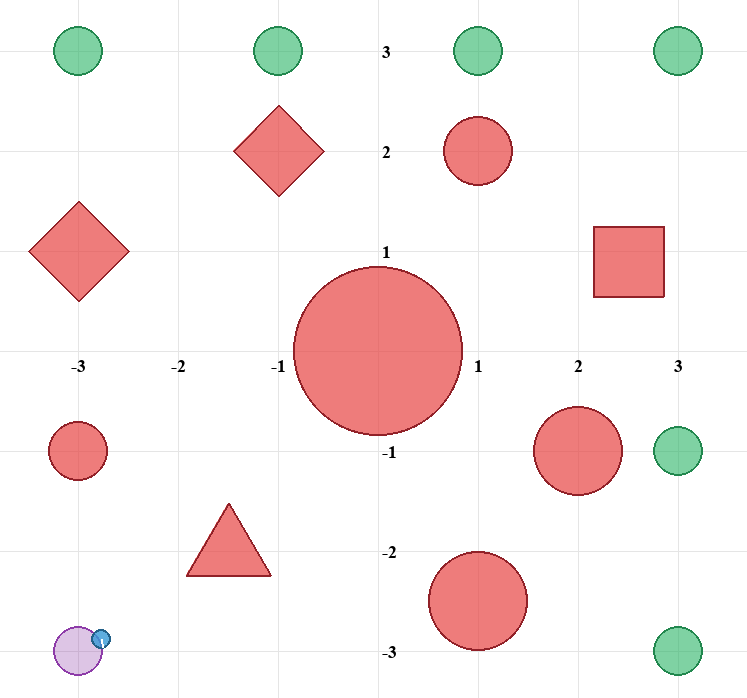}
        \caption{Goal-reaching task. At the start of each episode, a goal location is randomly selected from the set of green circles in the figure.}
        \label{fig:sub2}
      \end{subfigure}
    
      \caption{Task environments used in simulation studies. In both environments, the robot (blue circle) is randomly initialized within the purple region. Red regions indicate solid obstacles that the robot must avoid.}
      \label{fig:env}
      \vspace*{-2mm}
\end{figure}
\noindent\textbf{Unicycle:}
The system state comprises the robot's xy-position and the heading, i.e., $
s_k = \begin{bmatrix} x_k & y_k & \theta_k \end{bmatrix}^\top $. The control input comprises the normalized linear and angular velocities: $a_k = \begin{bmatrix} \bar{v}_k & \bar{\omega}_k \end{bmatrix}^\top \in [-1,1]^2 .$ The system dynamics are given by
\begin{equation}
\begin{aligned}
&s_{k+1} = s_k +
\begin{bmatrix}
v_{\max} \bar{v}_k \cos\theta_k \\
v_{\max} \bar{v}_k \sin\theta_k \\
\omega_{\max}  \bar{\omega}_k
\end{bmatrix} dt + d_k, \\
 &v_{\max} = 1.5 \, \text{m/s}, \quad \omega_{\max} = \pi \, \text{rad/s},\\
    &d_k \sim \mathcal{N}(0,\Sigma_d), \quad
\Sigma_d = dt^2\mathrm{diag}\!\left(0.03^2,0.03^2,0.05^2\right).
\end{aligned}
\end{equation}
\noindent\textbf{Ackermann car:} The system state is defined by the xy-position of the robot, heading of the robot, $\theta_k$,  and the steering angle, $\psi_k$, i.e., $s_k = \begin{bmatrix}
    x_k & y_k & \theta_k & \psi_k
\end{bmatrix}^\top$. The control input consists of the normalized linear velocity of the robot and the normalized steering angular velocity: $a_k = \begin{bmatrix} \bar{v}_k & \bar{\gamma}_k \end{bmatrix}^\top \in [-1,1]^2 .$ The system dynamics are given by
\begin{align}
&\!\!\!\!\!\!\!\!\!\!s_{k+1} = s_k +
\begin{bmatrix}
v_{\max} \bar{v}_k  \cos\theta_k \\
v_{\max} \bar{v}_k  \sin\theta_k \\
\frac{v_{\max} \bar{v}_k }{L}\tan(\psi_k) \\
\gamma_{\max} \bar{\gamma}_k
\end{bmatrix} dt + d_k,\\
    &\!\!\!\!\!\!\!\!\!\!v_{\max} = 1.5 \, \text{m/s}, \quad \gamma_{\max} = 2.0 \, \text{rad/s}, \quad L = 0.2 \, \text{m},\notag  \\
    &\!\!\!\!\!\!\!\!\!\!d_k \! \sim \! \mathcal{N}(0,\Sigma_d), \, \Sigma_d = dt ^2 \mathrm{diag}\!\left
    (0.03^2, 0.03^2, 0.05^2, 0.01^2 \right). \notag
\end{align}
The steering angle is clipped to satisfy $\psi_k \in  [-0.7, 0.7]$ after each state update to maintain physical realism.

\noindent\textbf{2D double integrator:} The system state is composed of the robot's xy-position and its velocities in the x and y direction, i.e., $s_k = \begin{bmatrix}
    x_k & y_k & v^x_k & v^y_k
\end{bmatrix}^\top$. The control input is the normalized accelerations in x and y  directions: $a_k = \begin{bmatrix}
        \bar{\alpha}_k^x & \bar{\alpha}_k^y
    \end{bmatrix}^\top \in [-1,1]^2 .$ The system dynamics are given by
\begin{align}
    &\!\!\!\!\!\!\!\!\!\!s_{k+1} = s_k + \begin{bmatrix}
        v^x_k \\ v^y_k \\ \alpha^x_{\max}\bar{\alpha}_k^x \\  \alpha^y_{\max}\bar{\alpha}_k^y
    \end{bmatrix} dt + d_k,\\
    &\!\!\!\!\!\!\!\!\!\!\alpha^x_{\max} = \alpha^y_{\max} = 3\, \text{m/}\text{s}^2, \quad d_k \sim \mathcal{N}(0, \Sigma_d), \notag \\
    &\!\!\!\!\!\!\!\!\!\! \Sigma_d = dt^2 \mathrm{diag}\left(0.03^2, 0.03^2, 0.1^2, 0.1^2\right). \notag
\end{align}
The robot speed is clipped to $1.5\, \text{m/s}$ after each state update.

Our tasks include circular path-following and goal-reaching tasks. In both, the robot is modeled as a circle of radius 0.1 m, and must avoid collisions with obstacles. Collision states are labeled unsafe in the simulation. Figure~\ref{fig:env} illustrates both environments. In the goal-reaching task, the goal is randomly selected from six circular goal regions of radius 0.25 m. The reward is the decrease in distance to the goal. In this task, the unit vector pointing toward the goal in the robot’s frame and $\min(d_\mathrm{goal}/8,\, 1)$, where $d_\mathrm{goal}$ is the distance to the goal, are appended to the state vector. An episode terminates when the robot reaches the goal or after $2000$ time steps.

The circular path-following task is inspired by the circle task in \cite{cpo}. In this task, the robot receives the maximum reward when following a circular path with a radius of 1.5 m, centered at the origin, in a counterclockwise direction at the highest speed permitted by its dynamics. Obstacle columns at $x = 1.25$ and $x = -1.25$ prevent the robot from crossing these lines. The reward is defined as
\begin{equation}
    r(s_k,a_k, d_k) = \frac{-v^x_{k+1} y_{k+1} + v^y_{k+1} x_{k+1}}{\left(1 + |\rho_{k+1}^{\text{robot}} - 1.5|\right) \rho_{k+1}^{\text{robot}}}
\end{equation}
where $\rho_{k+1}^{\text{robot}}$ denotes the distance of the robot from the origin at time step $k+1$, and $v^x_{k+1}$ and $v^y_{k+1}$ denote the robot's velocities along the $x$ and $y$ directions at time step $k+1$, respectively. An episode lasts $1000$ time steps in this task.

\begin{figure}
    \centering
\includegraphics[width=0.92\linewidth]{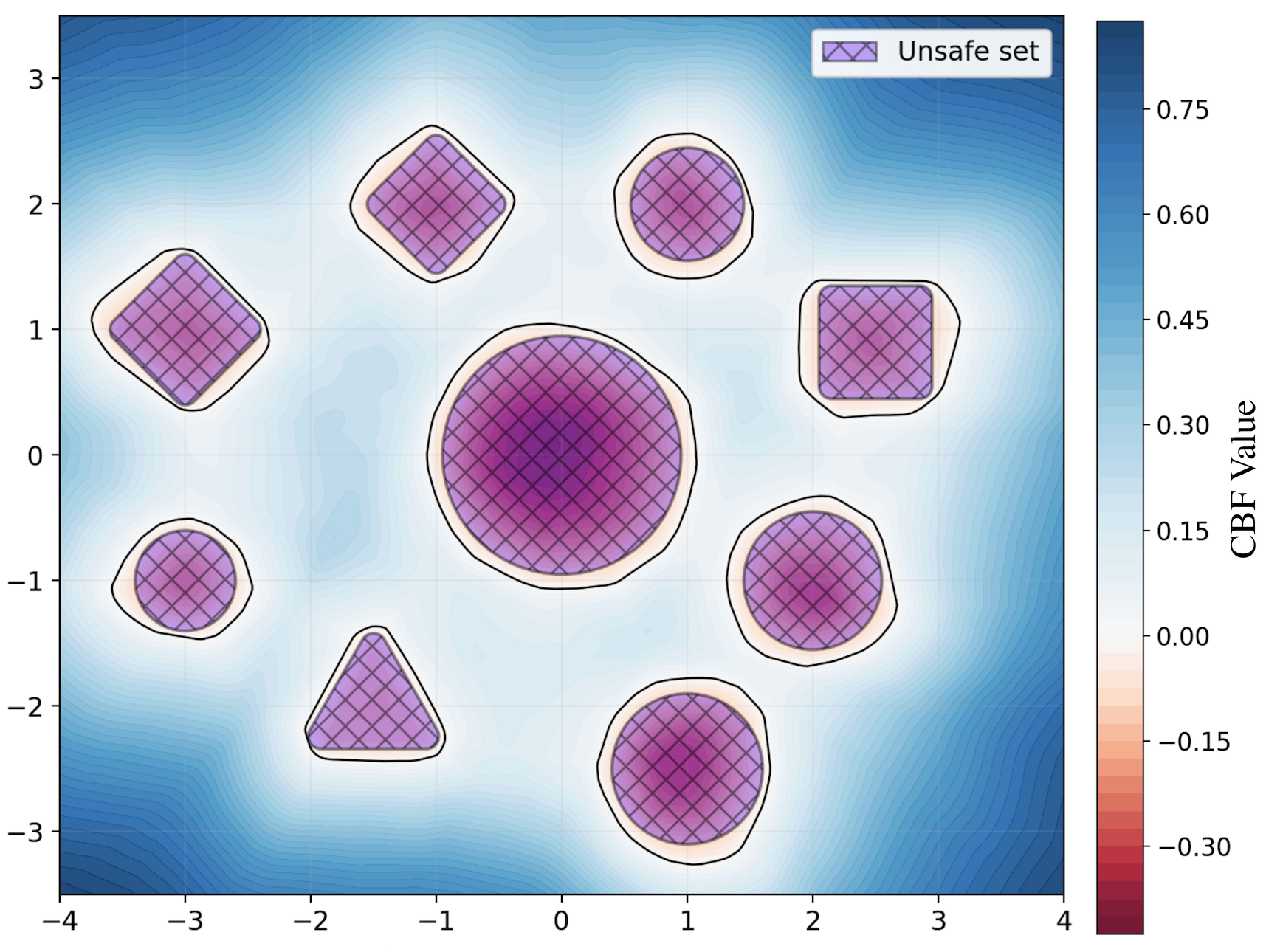}
    \caption{Trained CBF for a unicycle in the goal-reaching environment. The purple hatched regions indicate the unsafe states (obstacle inflated by the robot radius) in the environment. The regions enclosed by the black contour correspond to the unsafe set predicted by the trained CBF.}
    \label{fig:cbfs}
\end{figure}
\begin{figure}
    \centering
\includegraphics[width=0.85\linewidth]{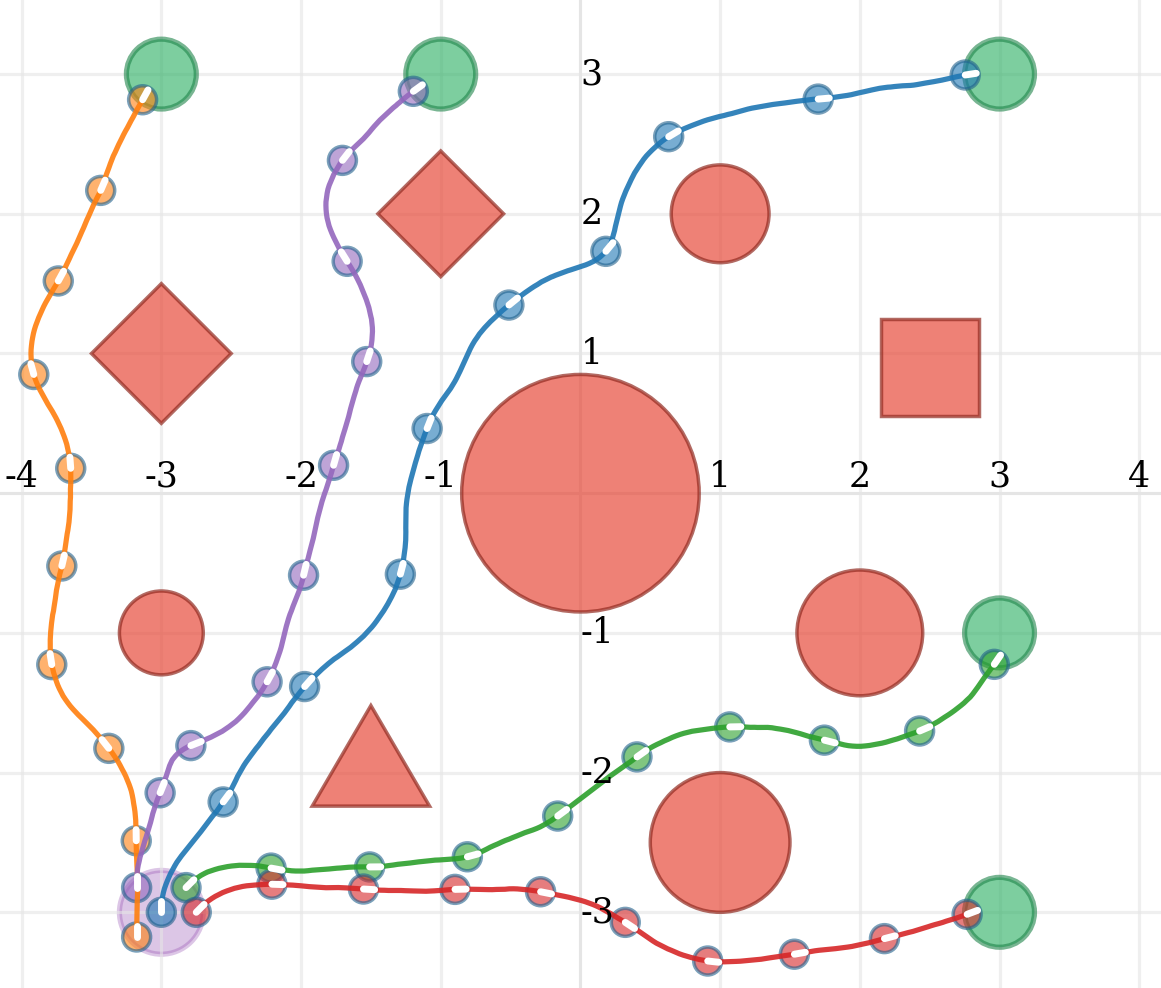}
    \caption{Examples of unicycle paths in the goal-reaching task.}
    \label{fig:unicycle_goal_reaching}
    \vspace*{-5mm}
\end{figure}

We employ a 36-beam LiDAR with a maximum range of 5 meters to implement the safety sensor that detects the safe and unsafe states in the robot's vicinity. The safety sensor first transforms the LiDAR point cloud into the global frame. For first-order systems, the safety sensor checks intersections between the point cloud and the robot’s collision body at sampled states in the LiDAR sensing horizon. If a collision occurs at a sampled state, it is identified as unsafe; otherwise, it is identified as safe. For the 2D double integrator, the safety sensor accounts for the velocity components of the sampled states relative to each LiDAR hit point. For each sampled state and each LiDAR hit point, it computes the velocity component along the direction from the state to that point and inflates the collision body in proportion to the square of this approaching component. If the robot is moving away (negative projection), no inflation is applied. It then checks for intersections between the point cloud and this per-(state, point) inflated collision body at the sampled states.

We compare PECTS with state-of-the-art safe RL methods, including PPO-Lag~\cite{safety_gym}, CPO~\cite{cpo}, and CUP~\cite{cup}. Since PECTS is an MPC-based method, we also include an MPC-based baseline that augments PETS~\cite{pets_paper} with a learned neural safety classifier trained on safety-sensor outputs, denoted as PETS+SC. Specifically, we incorporate the classifier output into the PETS MPC objective by adding a large negative penalty (-1000) for unsafe states. This penalty dominates the per-step reward scale in all tasks and effectively acts as a hard state constraint. We train the model-based methods (PECTS and PETS+SC) for 500 episodes, compared to 5000 for the other baselines, due to their higher sample efficiency.

The performance of the trained policies is compared in Table~\ref{tab:performances}. As shown, PECTS is the only method that satisfies the strict safety constraints in all experiments. To identify a strictly safe policy with the best performance, we define the best-performing method as the one with the highest success rate or average episode reward among the safest methods. Under this criterion, PECTS performs the best in 4 out of 6 tasks. While some baselines can exceed PECTS on some tasks when they are safe, they violate constraints in at least one other setting, and therefore do not provide a uniformly feasible solution under strict safety requirements.

\begin{figure*}
    \centering
\includegraphics[width=0.95\linewidth]{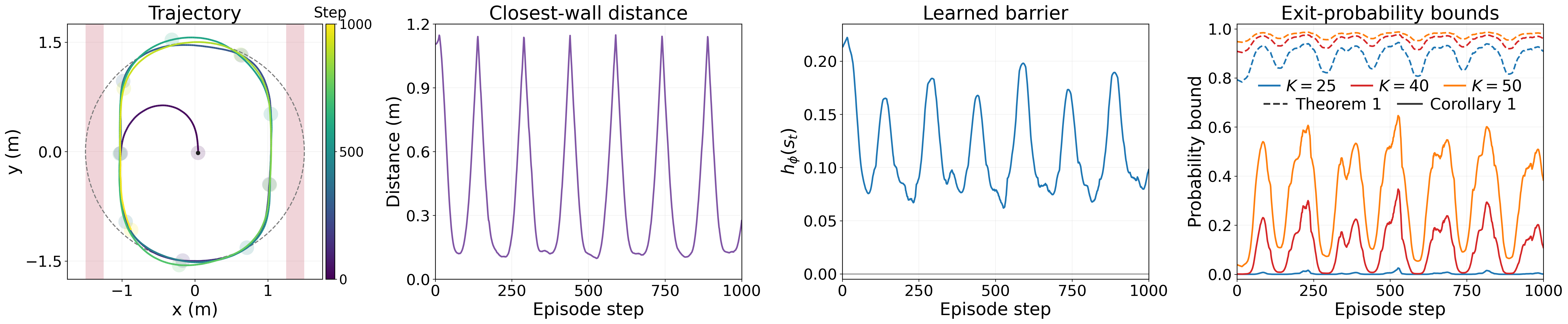}
    \caption{An episode of the Ackermann car circle-following experiment. From left to right: the executed trajectory, minimum clearance between the robot's circular boundary and the nearest wall, learned barrier values, and exit-probability bounds for horizons $K \in \{25,40,50\}$, evaluated along the trajectory. The bounds are calculated using the true disturbance distribution. We tighten the bounds from Theorem~\ref{thrm:1} by replacing $\sigma\sqrt{K}$ with $\sigma\sqrt{S_K}$, using the variance estimate in \eqref{eq:variance_est}.}
    \label{fig:bound_fig}
\end{figure*}
\begin{table*}
\vspace*{-4mm}
\centering
\caption{Performance of the trained policies. \emph{Ep. Rew.} denotes the episode return (mean $\pm$ std), \emph{Success (\%)} the task completion rate, and \emph{Safe (\%)} the percentage of constraint-satisfying episodes. For each method and each task, five policies are trained with different random seeds. Each policy is evaluated on the same 200 episodes for its task, and the reported metrics summarize the overall statistics. For each task, the best result is determined by first selecting methods with the highest Safe (\%), and then choosing the one with the highest Ep. Reward (or Success (\%)). The best result is highlighted in bold. }
\resizebox{\textwidth}{!}{%
\label{tab:performances}
\scalebox{0.9}{
\begin{tabular}{@{}lcccccccccccc@{}}
\toprule
\multirow{3}{*}{} & \multicolumn{4}{c}{\textbf{Unicycle}} & \multicolumn{4}{c}{\textbf{Ackermann Car}} & \multicolumn{4}{c}{\textbf{2D Double Integrator}} \\ \cmidrule(lr){2-5} \cmidrule(lr){6-9} \cmidrule(l){10-13}
 & \multicolumn{2}{c}{Circle} & \multicolumn{2}{c}{Goal Reaching} & \multicolumn{2}{c}{Circle} &  \multicolumn{2}{c}{Goal Reaching} & \multicolumn{2}{c}{Circle} & \multicolumn{2}{c}{Goal Reaching} \\ \cmidrule(lr){2-3} \cmidrule(lr){4-5} \cmidrule(lr){6-7} \cmidrule(lr){8-9} \cmidrule(lr){10-11} \cmidrule(lr){12-13}
 & Ep. Rew. & Safe (\%) & Success (\%) & Safe (\%) & Ep. Rew. & Safe (\%) & Success (\%) & Safe (\%) & Ep. Rew. & Safe (\%) & Success (\%) & Safe (\%) \\ \midrule
PECTS     & ${1195}{\scriptstyle {\pm 45}}$ & $100$& $\textbf{99.9}$ & $\textbf{100}$ &  ${\textbf{1154}}{\scriptstyle \pm \textbf{16}}$ & $\textbf{100}$ & $\textbf{98.0}$ & $\textbf{100}$ & ${1107}{\scriptstyle \pm 15}$ & $100$ & $\textbf{100}$ & $\textbf{100}$ \\
PETS + SC & $\textbf{{1275}}{\scriptstyle \pm \textbf{10}}$ &$\textbf{100}$& $90.8$ & $100$ & ${1239}{\scriptstyle \pm 104}$ & $89.5$ & $77.4$ & $100$ & ${1149}{\scriptstyle \pm 10}$ & $100$ & $99.5$ & $100$\\
PPO-Lag    & $1234{\scriptstyle \pm 125}$ & $97.9$ & $89.8$ & $94.5$ & ${1262}{\scriptstyle \pm 48}$ & $94$ & $3.4$ & $84.8$ & ${\textbf{1270}}{\scriptstyle \pm \textbf{3}}$ & $\textbf{100}$ & $99.3$ & $86.4$ \\
CPO       & $528{\scriptstyle \pm 285}$ & $94.7$ & $44.2$ & $99.7$ & ${266}{\scriptstyle \pm 218}$ & $100$ & $0.3$ & $99.5$ & ${289}{\scriptstyle \pm 400}$ & $99.9$ & $100$ & $98.4$\\
CUP       &  ${1137}{\scriptstyle \pm 246}$ & $76.6$ & $84.1$ & $91$ & ${1267}{\scriptstyle \pm 22}$ & $85.3$  & $4.5$ & $82.6$  & ${1240}{\scriptstyle \pm 42}$ & $80$ & $98$ & $81.4$ \\ \bottomrule
\end{tabular}%
}
}
\vspace*{-3mm}
\end{table*}

Figure~\ref{fig:cbfs} shows a trained CBF for the goal-reaching unicycle environment. The CBF accurately predicts all unsafe states as unsafe, which is critical for a safety-critical task. The boundary is mildly conservative, leaving a safety margin around obstacles, driven by the limited coverage of boundary states under the RL setting and by our safety-oriented hyperparameter choices ($2\lambda_1 = \lambda_2 = 2$ in (\ref{eq:cbf_loss}) and $\epsilon_+ <\epsilon_-$ in (\ref{eq:epsilons})). Beyond classification performance, another important consideration in CBF learning is feasibility, namely, whether there exists a control input that satisfies the CBF condition for all states in the domain of interest. Although infeasibility arises in early training episodes of PECTS, triggering recovery mode, it becomes increasingly rare as training progresses, and we did not observe infeasibility during the final evaluation of any task. This provides evidence consistent with the learned CBFs being feasible over the task-relevant subset of the state space. Figure~\ref{fig:unicycle_goal_reaching} shows a PECTS-controlled unicycle avoiding obstacles while progressing toward the goal, using the CBF in Figure~\ref{fig:cbfs}.

Figure~\ref{fig:bound_fig} shows an Ackermann car circle-following episo-de, together with the closest-wall distance, learned barrier value, and $K$-step exit-probability bounds. For Corollary~\ref{corollary:1}, we choose candidate cutoffs by setting the tail contribution $K\exp(-\delta^2/(2\sigma^2))$ equal to a prescribed budget $\varepsilon$, which gives $\delta_\varepsilon=\sigma\sqrt{2\log(K/\varepsilon)}$. For each horizon $K$ and each visited state, we evaluate the full bound using the fixed budgets $\varepsilon\in\{0.1,0.05,0.01,0.005, \allowbreak 0.001,0.0001\}$ and report the smallest value, which remains a valid upper bound because the corollary holds for every $\delta>0$. The barrier value tracks the closest-wall distance. The bounds from Theorem~\ref{thrm:1} remain close to one throughout the episode. Corollary~\ref{corollary:1} provides substantially smaller bounds than Theorem~\ref{thrm:1}. It yields bounds close to $0$ for $K=25$ while its bounds for $K\in \{40, 50\}$ increase when the robot is close to the unsafe region. The reduced informativeness at longer horizons is partly due to the factor $\kappa^K$ in the bound computation. This factor attenuates the contribution of the current barrier value $h_\phi(s_t)$. No safety violations were observed in the evaluation episodes, motivating further investigation of tighter bounds with less conservative horizon dependence.

\vspace*{-2mm}
\section{Conclusion}
In this paper, we introduced PECTS, an MPC-based safe RL approach that enforces probabilistic safety under model uncertainty. PECTS jointly learns stochastic dynamics using PNNs and CBFs using Lipschitz-bounded neural networks, and integrates the learned CBF constraints into an MPC formulation. Across a range of simulated safety-critical tasks and robot models, PECTS achieved stronger safety performance while maintaining competitive control performance relative to baseline methods. One limitation of PECTS is its reliance on complex safety sensor implementations for high-order systems. Future work will focus on enriching the framework with high-order CBFs to address this limitation and on validating PECTS in real-world robot experiments. Moreover, investigating tighter exit-probability bounds with less conservative horizon dependence remains a direction for future theoretical work.

\printcredits

\section*{Declaration of competing interest}
The authors declare that they have no known competing financial interests or personal relationships that could have appeared to influence the work reported in this paper.

\section*{Data availability}
Data will be made available on request.

\section*{Acknowledgements}
This work was supported in part by the New York University Abu Dhabi (NYUAD) Center for Artificial Intelligence and Robotics (CAIR), funded by Tamkeen under the NYUAD Research Institute Award CG010.

\bibliographystyle{model1-num-names}

\bibliography{references}

@ARTICLE{safe_rl_review,
  author={Gu, Shangding and Yang, Long and Du, Yali and Chen, Guang and Walter, Florian and Wang, Jun and Knoll, Alois},
  journal={{IEEE} Transactions on Pattern Analysis and Machine Intelligence}, 
  title={A Review of Safe Reinforcement Learning: Methods, Theories, and Applications}, 
  year={2024},
  volume={46},
  number={12},
  pages={11216-11235}}

@inproceedings{discrete_cbf,
  title={Discrete control barrier functions for safety-critical control of discrete systems with application to bipedal robot navigation},
  author={Agrawal, Ayush and Sreenath, Koushil},
  booktitle={Proc. Robotics: Science and Systems},
  volume={13},
  pages={1--10},
  year={2017},
  month={July},
  address={Cambridge, MA, US}
}

@inproceedings{discrete_cbf_mpc,
  title={Safety-critical model predictive control with discrete-time control barrier function},
  author={Zeng, Jun and Zhang, Bike and Sreenath, Koushil},
  booktitle={Proc. American Control Conference},
  pages={3882--3889},
  year={2021},
  month={May},
  address={New Orleans, LA, US}
}

@article{stochastic_discrete_cbf,
  title={Bounding stochastic safety: Leveraging {Freedman}’s inequality with discrete-time control barrier functions},
  author={Cosner, Ryan K and Culbertson, Preston and Ames, Aaron D},
  journal={{IEEE} Control Systems Letters},
  volume={8},
  pages={1937--1942},
  year={2024},
  note={Extended version available at arXiv:2403.05745}
}

@inproceedings{pets_paper,
 author = {Chua, Kurtland and Calandra, Roberto and McAllister, Rowan and Levine, Sergey},
 booktitle = {Proc. Advances in Neural Information Processing Systems},
 title = {Deep Reinforcement Learning in a Handful of Trials using Probabilistic Dynamics Models},
 pages = {4759–4770},
 volume = {31},
 year = {2018},
 month = {December},
 address = {Montreal, QC, Canada}
}

@article{pineda2021mbrl,
  title={{MBRL-Lib}: A Modular Library for Model-based Reinforcement Learning},
  author={Pineda, Luis and Amos, Brandon and Zhang, Amy and Lambert, Nathan O and Calandra, Roberto},
  journal={arXiv preprint arXiv:2104.10159},
  year={2021}
}

@ARTICLE{model_free_rl_cbf,
  author={Yang, Yujie and Jiang, Yuxuan and Liu, Yichen and Chen, Jianyu and Li, Shengbo Eben},
  journal={{IEEE} Robotics and Automation Letters}, 
  title={Model-Free Safe Reinforcement Learning Through Neural Barrier Certificate}, 
  year={2023},
  volume={8},
  number={3},
  pages={1295-1302}}

@ARTICLE{emam_safe_rl,
  author={Emam, Yousef and Notomista, Gennaro and Glotfelter, Paul and Kira, Zsolt and Egerstedt, Magnus},
  journal={{IEEE} Robotics and Automation Letters}, 
  title={Safe Reinforcement Learning Using Robust Control Barrier Functions}, 
  year={2025},
  volume={10},
  number={3},
  pages={2886-2893}}

@INPROCEEDINGS{airaldi2025probabilistically,
  author={Airaldi, Filippo and Schutter, Bart De and Dabiri, Azita},
  booktitle={Proc. Conference on Decision and Control}, 
  title={Probabilistically safe and efficient model-based reinforcement learning}, 
  year={2025},
  month={December},
  volume={},
  number={},
  pages={5853-5860},
  address={Rio de Janeiro, Brazil}
  }

@InProceedings{pmlr-v202-wang23as,
  title = 	 {Enforcing Hard Constraints with Soft Barriers: Safe Reinforcement Learning in Unknown Stochastic Environments},
  author =       {Wang, Yixuan and Zhan, Simon Sinong and Jiao, Ruochen and Wang, Zhilu and Jin, Wanxin and others},
  booktitle = 	 {Proc. International Conference on Machine Learning},
  pages = 	 {36593--36604},
  year = 	 {2023},
  volume = 	 {202},
  month = 	 {July},
  address={Honolulu, HI},
}

@INPROCEEDINGS{sabouni_safe_rl,
  author={Sabouni, Ehsan and Sabbir Ahmad, H.M. and Giammarino, Vittorio and Cassandras, Christos G. and Paschalidis, Ioannis Ch. and Li, Wenchao},
  booktitle={Proc. Conference on Decision and Control},
  title={Reinforcement Learning-based Receding Horizon Control using Adaptive Control Barrier Functions for Safety-Critical Systems},
  month={December},
  year={2024},
  pages={401-406},
  address={Milan, Italy}
  }

@inproceedings{NEURIPS2021_d71fa38b,
 author = {Luo, Yuping and Ma, Tengyu},
 booktitle = {Proc. Advances in Neural Information Processing Systems},
 pages = {25621--25632},
 title = {Learning Barrier Certificates: Towards Safe Reinforcement Learning with Zero Training-time Violations},
 volume = {34},
 year = {2021},
 month={December},
 address={}
}

@INPROCEEDINGS{model_free_rl_clbf_new,
  author={Du, Desong and Han, Shaohang and Qi, Naiming and Ammar, Haitham Bou and Wang, Jun and Pan, Wei},
  booktitle={Proc. International Conference on Robotics and Automation}, 
  title={Reinforcement Learning for Safe Robot Control using Control Lyapunov Barrier Functions}, 
  year={2023},
  pages={9442-9448},
  month={May},
  address={London, UK}}

@INPROCEEDINGS{bolun_cbf_learning,
  author={Dai, Bolun and Krishnamurthy, Prashanth and Khorrami, Farshad},
  booktitle={Proc. Conference on Decision and Control}, 
  title={Learning a Better Control Barrier Function}, 
  year={2022},
  month={December},
  address={Cancun, Mexico},
  volume={},
  number={},
  pages={945-950},}

@INPROCEEDINGS{robey_cbf_learning,
  author={Robey, Alexander and Hu, Haimin and Lindemann, Lars and Zhang, Hanwen and Dimarogonas, Dimos V. and Tu, Stephen and Matni, Nikolai},
  booktitle={Proc. Conference on Decision and Control}, 
  title={Learning Control Barrier Functions from Expert Demonstrations}, 
  year={2020},
  month={December},
  volume={},
  number={},
  pages={3717-3724},
  address={Jeju, Korea}
  }

@InProceedings{lipschitz_nn,
  title = 	 {Direct Parameterization of {L}ipschitz-Bounded Deep Networks},
  author =       {Wang, Ruigang and Manchester, Ian},
  booktitle = 	 {Proc. International Conference on Machine Learning},
  pages = 	 {36093--36110},
  year = 	 {2023},
  month = 	 {July},
  volume = 	 {202},
  address={Honolulu, HI},
}

@INPROCEEDINGS{song_safe_rl,
  author={Song, Lixing and Ferderer, Luke and Wu, Shaoen},
  booktitle={Proc. International Conference on Machine Learning and Applications}, 
  title={Safe Reinforcement Learning for {LiDAR}-based Navigation via Control Barrier Function}, 
  year={2022},
  month={December},
  volume={},
  number={},
  pages={264-269},
  address={Nassau, Bahamas}}

@InProceedings{PDDM,
  title = 	 {Deep Dynamics Models for Learning Dexterous Manipulation},
  author =       {Nagabandi, Anusha and Konolige, Kurt and Levine, Sergey and Kumar, Vikash},
  booktitle = 	 {Proc. Conference on Robot Learning},
  pages = 	 {1101--1112},
  year = 	 {2020},
  volume = 	 {100},
  month = 	 {October},
}

@InProceedings{cpo,
  title = 	 {Constrained Policy Optimization},
  author =       {Joshua Achiam and David Held and Aviv Tamar and Pieter Abbeel},
  booktitle = 	 {Proc. International Conference on Machine Learning},
  pages = 	 {22--31},
  year = 	 {2017},
  volume = 	 {70},
  month = 	 {August},
  address= {Sydney, Australia}
}

@misc{safety_gym,
  title        = {Benchmarking Safe Exploration in Deep Reinforcement Learning},
  author       = {Ray, Alex and Achiam, Joshua and Amodei, Dario},
  year         = {2019},
  howpublished = {Preprint},
  url          = {https://cdn.openai.com/safexp-short.pdf}
}

@inproceedings{cup,
 author = {Yang, Long and Ji, Jiaming and Dai, Juntao and Zhang, Linrui and Zhou, Binbin and others},
 booktitle = {Proc. Advances in Neural Information Processing Systems},
 pages = {9111--9124},
 title = {Constrained Update Projection Approach to Safe Policy Optimization},
 volume = {35},
 year = {2022},
 month={November},
 address={New Orleans, LA}
 
}

@INPROCEEDINGS{9993190,
  author={Tao, Chuyuan and Yoon, Hyung-Jin and Kim, Hunmin and Hovakimyan, Naira and Voulgaris, Petros},
  booktitle={Proc. Conference on Decision and Control}, 
  title={Path Integral Methods with Stochastic Control Barrier Functions}, 
  year={2022},
  month={December},
  volume={},
  number={},
  pages={1654-1659},
  address={Cancun, Mexico}}

@INPROCEEDINGS{11312129,
  author={Rabiee, Pedram and Hoagg, Jesse B.},
  booktitle={Proc. Conference on Decision and Control}, 
  title={Guaranteed-Safe {MPPI} Through Composite Control Barrier Functions for Efficient Sampling in Multi-Constrained Robotic Systems}, 
  year={2025},
  volume={},
  number={},
  pages={5515-5520},
  month={December},
  address={Rio de Janeiro, Brazil}
   }

@INPROCEEDINGS{YinJ-RSS-25, 
    AUTHOR    = {Ji Yin AND Oswin So AND Eric Yang Yu AND Chuchu Fan AND Panagiotis Tsiotras}, 
    TITLE     = {{Safe Beyond the Horizon: Efficient Sampling-based {MPC} with Neural Control Barrier Functions}}, 
    BOOKTITLE = {Proc. Robotics: Science and Systems}, 
    YEAR      = {2025}, 
    ADDRESS   = {Los Angeles, CA}, 
    MONTH     = {June}
}

@article{KAYPAK2026105263,
title = {Safe multi-robotic arm interaction via {3D} convex shapes},
journal = {Robotics and Autonomous Systems},
volume = {196},
pages = {105263},
year = {2026},
author = {Ali Umut Kaypak and Shiqing Wei and Prashanth Krishnamurthy and Farshad Khorrami},
}

@INPROCEEDINGS{bolun_2,
  author={Dai, Bolun and Huang, Heming and Krishnamurthy, Prashanth and Khorrami, Farshad},
  booktitle={Proc. American Control Conference}, 
  title={Data-Efficient Control Barrier Function Refinement}, 
  year={2023},
  month={May},
  volume={},
  number={},
  pages={3675-3680},
  address={San Diego, CA, US}
  }

@article{hoeffding_ineq,
title = {Hoeffding’s inequality for supermartingales},
journal = {Stochastic Processes and their Applications},
volume = {122},
number = {10},
pages = {3545-3559},
year = {2012},
author = {Xiequan Fan and Ion Grama and Quansheng Liu},
}

@incollection{concentration_ineq,
    author = {Boucheron, Stéphane and Lugosi, Gábor and Massart, Pascal},
    title = {Logarithmic {Sobolev} Inequalities},
    booktitle = {Concentration Inequalities: A Nonasymptotic Theory of Independence},
    publisher = {Oxford University Press},
    year = {2013},
    month = {02}
}



\end{document}